\documentclass[]{article}

\usepackage{amsmath, amssymb, dsfont, mathrsfs,amsthm}
\usepackage{algorithm}
\usepackage[noend]{algpseudocode}
\usepackage{float}
\usepackage{mdframed}
\usepackage{appendix}

\usepackage[hidelinks]{hyperref}
\usepackage{graphicx}

\usepackage{caption}
\usepackage[margin=1in]{geometry}

\allowdisplaybreaks

\title{Simple Codes and Sparse Recovery with Fast Decoding
}

\author{Mahdi Cheraghchi\thanks{EECS Department, University of Michigan, Ann Arbor, MI, USA. Email: \texttt{mahdich@umich.edu}.} \and Jo\~ao Ribeiro\thanks{Computer Science Department, Carnegie Mellon University, Pittsburgh, PA, USA. Email: \texttt{jlourenc@cs.cmu.edu}.} \thanks{This work was mainly performed while the authors were with the Department of Computing, Imperial College London, UK.
}
}

\date{}

\newcommand{\cC}{\mathcal{C}}

\newcommand{\cE}{\mathcal{E}}
\newcommand{\cX}{\mathcal{X}}

\newcommand{\cS}{\mathcal{S}}

\newcommand{\F}{\mathds{F}}

\newcommand{\supp}{\mathsf{supp}}

\newcommand{\eps}{\epsilon}

\newtheorem{thm}{Theorem}

\newtheorem{lem}[thm]{Lemma}

\newtheorem{coro}[thm]{Corollary}

\newtheorem{defn}[thm]{Definition}
\newtheorem{remark}[thm]{Remark}

\newcommand{\cL}{\mathcal{L}}
\newcommand{\cR}{\mathcal{R}}
\newcommand{\poly}{\textrm{poly}}
\newcommand{\GF}{\textrm{GF}}

\newcommand{\bin}{\mathsf{bin}}

\newenvironment{eqn}{\[}{\]}

\let\originalleft\left
\let\originalright\right
\renewcommand{\left}{\mathopen{}\mathclose\bgroup\originalleft}
\renewcommand{\right}{\aftergroup\egroup\originalright}

\usepackage{url}
\begin{document}

\maketitle

\begin{abstract}
		Construction of error-correcting codes achieving a designated 
		minimum distance parameter is a central problem in coding
		theory.
		In this work, we study a very simple construction
		of binary linear codes that correct a given number of errors $K$. Moreover, we design a simple, nearly 
		optimal syndrome decoder for the code as well. 
		The
		running time of the decoder is only logarithmic
		in the block length of the code, and nearly linear
		in the number of errors $K$.
		This decoder can be applied to exact for-all sparse recovery over any field, improving upon previous results with the same number of measurements.
		Furthermore,
		computation of the syndrome from a received word
		can be done in nearly linear time in the block length.
        We also demonstrate an application of these techniques in
		non-adaptive group testing, and construct simple explicit measurement
		schemes with $O(K^2 \log^2 N)$ tests and
		$O(K^3 \log^2 N)$ recovery time for identifying
		up to $K$ defectives in a population of size $N$.
\end{abstract}

\section{Introduction}\label{sec:intro}

	The problem of constructing low-redundancy codes with practical decoding algorithms that handle a prescribed number of adversarial errors has been extensively studied in coding theory. We distinguish between two standard decoding settings for linear codes: \emph{Syndrome decoding}, where one has access to the syndrome of the corrupted codeword, and \emph{full decoding}, where one has access to the corrupted codeword itself. In both cases, the goal is to return the error pattern.
	
	The syndrome holds extra pre-computed information about the corrupted codeword. As a result, we expect to be able to perform syndrome decoding much faster than full decoding. In fact, while full decoding has complexity at least linear in the block length of the code, syndrome decoding can be potentially accomplished in time \emph{sublinear} in the block length.
	Syndrome decoding is important for various reasons: In many cases, the most efficient way we have of performing full decoding for a given linear code is to first compute the syndrome from the corrupted codeword, and then run a syndrome decoding algorithm.
	Furthermore, as we shall see, syndrome decoding is connected to other widely studied recovery problems.
	
	Two examples of families of high-rate linear codes with good decoding guarantees are the classical BCH codes, which are widely used in practice, and expander codes, introduced by Sipser and Spielman~\cite{SS96}. The properties of BCH codes derive from the theory of polynomials over large finite fields. On the other hand, the guarantees behind expander codes follow from the combinatorial properties of the underlying expander graphs. Due to their combinatorial nature, expander codes have simple decoders, while decoding BCH codes requires  algorithms which perform arithmetic over large fields. However, while BCH codes support sublinear syndrome decoding~\cite{DRS04,DORS06}, no such sublinear syndrome decoders are known for expander codes.
	
	Syndrome decoding of linear codes can be interpreted as sparse recovery over a finite field. 
	In \emph{exact for-all} sparse recovery, we aim to recover all sparse vectors $x$ from $Wx$, where $W$ is a measurement matrix (which may be sampled with high probability from some distribution). The goal is to minimize recovery time and number of measurements (i.e., rows of $W$) with respect to the sparsity and length of the vectors. We say a for-all sparse recovery scheme is \emph{approximate} if it allows recovery (within some error) of the best sparse approximation of arbitrary vectors. 
	There has been significant interest in developing combinatorial algorithms for sparse recovery. Unlike their geometric counterparts, such procedures have the advantage of supporting sublinear time recovery. Furthermore, while sparse recovery is normally studied over the reals, such algorithms can usually be modified to perform sparse recovery over any field.
	
 	In this work, we study a simple combinatorial construction of high-rate binary linear codes that support nearly optimal syndrome decoding. 
 	While we present our decoding algorithms over $\F_2$, our syndrome decoder can be adapted to work over any field, improving upon previous combinatorial algorithms for exact for-all sparse recovery  with the same number of measurements.
 	This adaptation to arbitrary fields is discussed in more detail at the end of Section~\ref{sec:detsynrec}.
 	We remark that the parity-check matrix of our code remains $0$-$1$-valued even when we work over fields other than GF(2).
	
	A different sparse recovery problem that has been extensively studied is that of \emph{Non-Adaptive Group Testing} (NAGT). In this setting, our goal is to identify all defectives within a given population. To this end, we may perform tests by pooling items of our choice and asking whether there is a defective item in the pool. In NAGT, tests are fixed a priori, and so cannot depend on outcomes of previous tests. A main problem in this area consists in finding NAGT schemes supporting sublinear recovery time with few tests in the \emph{zero-error} regime, where it is required that the scheme always succeeds in recovering the set of defectives. 
	In the second part of this work, we present a zero-error NAGT scheme with few tests and a competitive sublinear time recovery algorithm.
	
	\subsection{Related Work}\label{sec:relatedwork}
	
	Sublinear time decoding of BCH codes was studied by Dodis, Reyzin, and Smith~\cite{DRS04,DORS06}, who gave a $\poly(K\log N)$ time syndrome decoder, where $N$ is the block length and $K$ is the maximum number of errors. 
	More precisely, according to~\cite[Proof of Lemma 1]{DORS06} their decoder asymptotically requires $\Theta(K\log^2 K\cdot\log\log K +K^2 \log N)$ multiplications over the extension field of order $N+1$.
	In turn, we can perform multiplications over this field in time $\Theta(\log N\cdot \log\log N)$ under a believable assumption~\cite{HV22}.
	Combining both results yields asymptotic syndrome decoding complexity
	\begin{equation*}
	    \Theta((K\log^2 K\cdot\log\log K +K^2 \log N)\cdot \log N\cdot \log\log N).
	\end{equation*}
	
	In contrast, the best syndrome decoders for expander codes run in time $O(DN)$, where $D$ is the left degree of the expander~\cite{JXHC09}. Full decoding of expander codes takes time $O(N)$ in the regime where $K=\Theta(N)$~\cite{SS96}, but it takes time $O(N\log N)$ when $K$ is small and we want the rate of the expander code to be large~\cite{JXHC09}.
	
	The work on combinatorial for-all sparse recovery algorithms was initiated by Cormode and Muthukrishnan~\cite{CM06}, and several others soon followed~\cite{GSTV06,GSTV07,BGIKS08} with improved recovery time and number of measurements. More recently, sublinear time combinatorial algorithms for approximate for-all sparse recovery with optimal number of measurements~\cite{GLPS17} or very efficient recovery~\cite{CI17} were given (both with strong approximation guarantees). We compare the results obtained by these works in the context of for-all sparse recovery with the result we obtain in this work in Section~\ref{sec:contributions}.
	
	The first zero-error NAGT schemes supporting sublinear recovery time with few tests were given independently by Cheraghchi~\cite{Che13} and Indyk, Ngo, and Rudra~\cite{INR10}. In~\cite{INR10}, the authors present an NAGT scheme which requires $T=O(K^2\log N)$ tests (which is order-optimal) and supports recovery in time $\textrm{poly}(K)\cdot T\log^2 T+O(T^2)$, where $N$ is the population size and $K$ is the maximum number of defectives. However, their scheme is only explicit when $K=O\left(\frac{\log N}{\log\log N}\right)$. 
	Cheraghchi~\cite{Che13} gives explicit schemes which require a small number of tests and handle a constant fraction of test errors. However, his schemes may output some false positives. While this can be remedied, the resulting recovery time will be worse.
	Later, Ngo, Porat, and Rudra~\cite{NPR11} also obtained explicit sublinear time NAGT schemes with near-optimal number of tests that are robust against test errors. In particular, they obtain schemes requiring $T=O(K^2\log N)$ tests with recovery time $\textnormal{poly}(T)$. 
Subsequently to the publication of a shortened version of the present work~\cite{CR19}, Cheraghchi and Nakos~\cite{CN20} constructed explicit NAGT schemes requiring $T=O(K^2\log N)$ tests and recovery time nearly linear in $T$.

	\subsection{Contributions and Techniques}\label{sec:contributions}
	
	Our binary linear codes are a \emph{bitmasked} version of expander codes.
	Roughly speaking, bitmasking a binary $M\times N$ matrix $W$ consists in replacing each entry in $W$ by the $\log N$-bit binary expansion of its corresponding column index, multiplied by that entry of $W$. This gives rise to a new bitmasked matrix of dimensions $M\log N\times N$. 
	The bitmasking technique has already been used in~\cite{CM06,GSTV06,GSTV07,BGIKS08,CI17} in several different ways to obtain sublinear time recovery algorithms for approximate sparse recovery. We provide a detailed explanation of this technique and its useful properties in Section~\ref{sec:sparserec}.
	
	The parity-check matrices of our codes are obtained by following the same ideas as~\cite{BGIKS08,CI17}: 
    We take the parity-check matrix of an expander code, bitmask it, and stack the two matrices. Note that these codes have a blowup of $\log N$, where $N$ is the block length, in the redundancy when compared to expander codes. However, we show that these codes support randomized and deterministic syndrome decoding in time 
    $O(K\log K\cdot \log N)$ under a random expander, where $K$ is the number of errors. We remark that this is within an $O(\log K)$ factor of the optimal recovery time for small $K$. In particular, this also leads to randomized and deterministic full decoders running in time $O(\log K\cdot N\log N)$.
    
    Our syndrome decoders can be made to work over any field with almost no modification. 
    As a result, we obtain a recovery algorithm for exact for-all sparse recovery over any field with nearly optimal recovery time $O(K\log K\cdot \log N)$ from $O(K\log^2 N)$ measurements. This improves upon the recovery time of previous schemes using the same number of measurements in the exact for-all sparse recovery setting~\cite{BGIKS08,CI17}. Although sublinear time recovery is possible with fewer measurements (the optimal number of measurements is $O(K\log(N/K))$), the dependency on $N$ in that case is generally worse than what we obtain, which is optimal. A detailed comparison between our work and previous results in the for-all sparse recovery setting can be found in Table~\ref{table:sparserecresults}.

    Our randomized decoders have several advantages over their deterministic counterparts that make them more practical. 
    First, the hidden constants in the runtime are smaller. 
    Second, the runtime is independent of the degree of the underlying expander. 
    As a result, we can instantiate our codes under explicit expanders with sub-optimal degree without affecting the decoding complexity. Third, the failure probability of the algorithm has a negligible effect on its runtime for large block lengths. Therefore, it can be set to an arbitrarily small constant of choice with limited effect on the runtime. 
    In particular, our randomized full decoder is faster than the expander codes decoder even under a random expander if $K$ is small.
	
        \begin{table}[h]
    \centering
\begin{tabular}{c|c|c|}
\cline{2-3}
                                & Recovery time & Approximate? (Y/N) \\
\hline
\multicolumn{1}{|c|}{\cite{CM06}}          &       $K^2\log^2 N$        & N                  \\ 
\hline
\multicolumn{1}{|c|}{\cite{GSTV06}} & $K\log^2 K\cdot \log^2 N$               & Y                  \\ \hline
\multicolumn{1}{|c|}{\cite{GSTV07}} & $K^2\cdot\textrm{polylog}(N)$            & Y                  \\ \hline
\multicolumn{1}{|c|}{\cite{BGIKS08}}          & $K\log K\cdot\log^2 N$              & N                  \\ \hline
\multicolumn{1}{|c|}{\cite{CI17}}          &  $K\log K\cdot\log^2 N$             & Y                  \\ \hline
\multicolumn{1}{|c|}{This work (see Remark~\ref{rem:readint})} &  $K\log K\cdot\log N$             & N                  \\ \hline
\multicolumn{1}{|c|}{$< K\log^2 N$ measurements~\cite{GLPS17}}          &  $\textrm{poly}(K\log N)$             & Y                  \\ \hline
\multicolumn{1}{|c|}{BCH codes~\cite{DORS06,HV22}}          &       $K^2\log^2 N\cdot \log\log N$        & N                  \\ \hline
\end{tabular}
\caption{Summary of best known previous results and the result obtained in this work on sublinear recovery in the for-all sparse recovery setting. Here, $K$ denotes the sparsity and $N$ the vector length. We omit the $O(\cdot)$ notation in recovery times for simplicity, and in the third column we distinguish between schemes that work for approximate sparse recovery (i.e., arbitrary vectors), versus those that work only in the exact for-all setting. Since it is not relevant for us, we do not distinguish between the approximation guarantees obtained in each work for approximate sparse recovery. 
All works except the last two rows require at least $O(K\log^2 N)$ measurements, which is the number of measurements our scheme requires. 
While sublinear decoding is possible in those cases, the dependency on $N$ is generally worse than what we obtain. For a more detailed description of such schemes, see~\cite[Table 1]{GLPS17} and~\cite[Table 1]{GSTV06}. As mentioned before, our measurement matrices remain $0$-$1$-valued even when we work over fields other than GF(2).}
\label{table:sparserecresults}
\end{table}
    
\begin{remark}\label{rem:readint}
    We assume that reading an integer from memory takes time $O(1)$. If instead we assume that reading an $L$-bit integer takes time $O(L)$, then we incur an extra $\log K$ factor in our deterministic recovery time, for a total runtime of $O(K\log^2 K\cdot \log N)$.
\end{remark}
	
	In the second part of our work, we present a simple, explicit zero-error NAGT scheme with recovery time $O(K^3\log^2 N)$ from $O(K^2\log^2 N)$ tests, where $N$ is the population size and $K$ is the maximum number of defectives. Such a scheme is obtained by bitmasking a \emph{disjunct} matrix. At a high-level, the recovery algorithm works as follows: First, we use the bitmask to obtain a small superset of the set of defectives. Then, we simply apply the naive recovery algorithm for general disjunct matrices to this superset to remove all false positives. The recovery time of this scheme is better than of those presented in Section~\ref{sec:relatedwork}, albeit we are an $O(\log N)$ factor away from the optimal number of tests. Moreover, unlike our scheme, those schemes make use of algebraic list-decodable codes and hence require sophisticated recovery algorithms with large constants. Finally, we note that the bitmasking technique has been used in a different way to obtain efficient NAGT schemes which recover a large fraction of defectives, or even all defectives, with high probability~\cite{LPR16}.

\subsection{Organization}

In Section~\ref{sec:prelims}, we introduce several concepts and results in coding and group testing that will be relevant for our work in later sections. Then, in Section~\ref{sec:sparserec} we present the bitmasking technique and its original application in sparse recovery. We present our code construction and the decoding algorithms in Section~\ref{sec:code}. Finally, our results on non-adaptive group testing can be found in Section~\ref{sec:grouptesting}.

\section{Preliminaries}\label{sec:prelims}

\subsection{Notation}

We denote the set $\{0,\dots,N-1\}$ by $[N]$. Given a vector $x$, we denote its support $\{i:x_i\neq 0\}$ by $\supp(x)$. We say a vector is $K$-sparse if $|\supp(x)|\leq K$. We index vectors and matrix rows/columns starting at $0$. Sets are denoted by calligraphic letters like $\cS$ and $\cX$. In general, we denote the base-$2$ logarithm by $\log$. Given a graph $G$ and a set of vertices $\cS$, we denote its neighborhood in $G$ by $\Gamma(\cS)$. For a matrix $W$, we denote its $i$-th row by $W_{i\cdot}$ and its $j$-th column by $W_{\cdot j}$.

\subsection{Unbalanced Bipartite Expanders}\label{sec:expanderscondensers}

In this section, we introduce unbalanced bipartite expander graphs. We will need such graphs to define our code in Section~\ref{sec:code}.

\begin{defn}[Bipartite Expander]\label{def:expander}
    A bipartite graph $G=(\cL,\cR,\cE)$ is said to be a \emph{$(D,K,\eps)$-bipartite expander} if every vertex $u\in \cL$ has degree $D$ (i.e., $G$ is left $D$-regular) and for every $\cS\subseteq \cL$ satisfying $|\cS|\leq K$ we have $|\Gamma(\cS)|\geq (1-\eps) D|\cS|$. 
    
    Such a graph is said to be \emph{layered} if we can partition $\cR$ into disjoint subsets $\cR_1,\dots,\cR_D$ with $|\cR_i|=|\cR|/D$ for all $i$ such that every $u\in \cL$ has degree $1$ in the induced subgraph $G_i=(\cL,\cR_i,\cE)$. We call such $i\in [D]$ \emph{seeds} and denote the neighborhood of $\cS$ in $G_i$ by $\Gamma_i(\cS)$.
    
    The graph $G$ can be defined by the function $C:\cL\times[D]\to\cR$ which maps $(u,i)\in\cL\times[D]$ to the neighbor of $u$ in $\cR_i$.
    For convenience, we also define $C_i=C(\cdot,i)$, which defines the subgraph $G_i$.
\end{defn}
Informally, we say that a bipartite expander graph is \emph{unbalanced} if $|\cR|$ is much smaller than $|\cL|$. The next lemma follows immediately from Markov's inequality.
\begin{lem}\label{lem:markov}
    Fix some $c>1$ and a set $\cS$ such that $|\cS|\leq K$, and let $G$ be a $(D,K,\eps)$-layered bipartite expander. Then, for at least a $(1-1/c)$-fraction of seeds $i\in[D]$, it holds that
    \begin{equation*}
        |\Gamma_i(\cS)|\geq (1-c\eps)|\cS|.
    \end{equation*}
\end{lem}

We will also need the following lemma that bounds the number of right vertices with a single neighbor in a given subset of left vertices.
\begin{lem}\label{lem:singleneighbor}
    Let $G=(\cL,\cR,\cE)$ be a layered bipartite graph.
    If $\cS\subseteq\cL$ satisfies $|\Gamma_i(\cS)|\geq(1-\delta)|\cS|$, then the number of right vertices in $\Gamma_i(\cS)$ with only one neighbor in $\cS$ (with respect to the subgraph $G_i$) is at least $(1-2\delta)|\cS|$.
\end{lem}
\begin{proof}
    Suppose that there are fewer than $(1-2\delta)|\cS|$ vertices in $\Gamma_i(\cS)$ with only one neighbor in $\cS$.
    Since $G_i$ has left-degree $1$, this means that there are more than $2\delta|\cS|$ vertices in $\cS$ which are adjacent to right vertices with degree at least $2$.
    Therefore, we can upper bound $|\Gamma_i(\cS)|$ as
    \begin{equation*}
        |\Gamma_i(\cS)|<(1-2\delta)|\cS|+\frac{1}{2}\cdot 2\delta|\cS| =(1-\delta)|cS|,
    \end{equation*}
    which contradicts our assumption.
\end{proof}

The next theorem states we can sample nearly-optimal layered unbalanced bipartite expanders with high probability.
\begin{thm}[\protect{\cite[Lemma 4.2]{CRVW02}}]\label{thm:expander}
   Given any $N$, $K$, and $\eps$, we can sample a layered $(D,K,\eps)$-bipartite expander graph $G=(\cL=[N],\cR=[MD],\cE)$ with high probability for $D=O\left(\frac{\log N}{\eps}\right)$ and $M=O\left(\frac{K}{\eps}\right)$.
\end{thm}
The graph in Theorem~\ref{thm:expander} can be obtained by sampling a random function $C\colon [N]\times [D]\to [M]$, and choosing $(x;(s,y))$ as an edge in $G$ if $C(x,s)=y$. 

\subsection{Coding Theory}

In this section, we define some basic concepts from coding theory that we use throughout our paper. We point the reader to~\cite{MS77} for a much more complete treatment of the topic.

Given an alphabet $\Sigma$, a \emph{code over $\Sigma$ of length $N$} is simply a subset of $\Sigma^N$. We will be focusing on the case where codes are \emph{binary}, which corresponds to the case $\Sigma=\{0,1\}$. For a code $\cC\subseteq\Sigma^N$, we call $N$ the \emph{block length} of $\cC$. If $|\Sigma|=q$, the \emph{rate} of $\cC$ is defined as
\begin{equation*}
    \frac{\log_q|\cC|}{N}.
\end{equation*}
We will study families of codes indexed by the block length $N$. We do not make this dependency explicit, but it is always clear from context.

If $\Sigma$ is a field, we say $\cC$ is a \emph{linear code} if $c_1+c_2\in\cC$ whenever $c_1,c_2\in\cC$. In other words, $\cC$ is a linear code if it is a vector space over $\Sigma$. To each linear code $\cC$ we can associate a unique \emph{parity-check matrix} $H$ such that $\cC=\ker H$. Given such a parity-check matrix $H$ and a vector $x\in\Sigma^N$, we call $Hx$ the \emph{syndrome} of $x$. Clearly, we have $x\in\cC$ if and only if its syndrome is zero.

\subsection{Group Testing}\label{sec:nagtpre}

As mentioned in Section~\ref{sec:intro}, in group testing we are faced with a set of $N$ items, some of which are defective. Our goal is to correctly identify all defective items in the set. To this end, we are allowed to test subsets, or pools, of items. The result of such a test is $1$ if there is a defective item in the pool, and $0$ otherwise. Ideally, we would like to use as few tests as possible, and have efficient algorithms for recovering the set of defective items from the test results.

In Non-Adaptive Group Testing (NAGT), all tests are fixed a priori, and so cannot depend on the outcome of previous tests. While one can recover the defective items with fewer tests using adaptive strategies, practical constraints preclude their use and make non-adaptive group testing relevant for most applications.

It is useful to picture the set of $N$ items as a binary vector $x\in\{0,1\}^N$, where the $1$'s stand for defective items. Then, the $T$ tests to be performed can be represented by a $T\times N$ test matrix $W$ satisfying
\begin{equation*}
    W_{ij}=\begin{cases}
        1,\textrm{ if item $j$ is in test $i$}\\
        0,\textrm{ else.}
    \end{cases}
\end{equation*}
The outcome of the $T$ tests, which we denote by $W\odot x$, corresponds to the bit-wise union of all columns corresponding to defective items. In other words, if $\cS$ denotes the set of defective items, we have
\begin{equation*}
    W\odot x =\bigvee_{j\in\cS} W_{\cdot j},
\end{equation*}
where the bit-wise union of two $N$-bit vectors, $x\vee y$, is an $N$-bit vector satisfying
\begin{equation*}
    (x\vee y)_i=\begin{cases}
        1,\textrm{ if $x_i=1$ or $y_i=1$}\\
        0,\textrm{ else.}
    \end{cases}
\end{equation*}

We say an NAGT scheme is \emph{zero-error} if we can always correctly recover the set of defectives. Zero-error NAGT schemes are equivalent to \emph{disjunct} matrices.
\begin{defn}[Disjunct matrix]\label{def:disjunct}
    A matrix $W$ is said to be \emph{$d$-disjunct} if the bit-wise union of any up to $d$ columns of $W$ does not contain any other column of $W$.
\end{defn}
The term \emph{contains} used in Definition~\ref{def:disjunct} is to be interpreted as follows: A vector $x$ is contained in a vector $y$ if $y_i=1$ whenever $x_i=1$, for all $i$. If we know there are at most $K$ defective items, taking the rows of a $K$-disjunct $T\times N$ matrix as the tests to be performed leads to a NAGT algorithm with $T$ tests and a simple $O(TN)$ recovery algorithm: An item is not defective if and only if it is part of some test with a negative outcome, so one can just check whether each item participates in a negative test.

As a result, much effort has been directed at obtaining better randomized and explicit\footnote{By an \emph{explicit} construction, we mean one in which we can construct the matrix in time polynomial in $N$.} constructions of $K$-disjunct matrices, with as few tests as possible with respect to number of defectives $K$ and population size $N$. The current best explicit construction of a $K$-disjunct matrix due to Porat and Rothschild~\cite{PR08} requires $O(K^2\log N)$ rows (i.e., tests), while a probabilistic argument shows that it is possible to sample $K$-disjunct matrices with $O(K^2\log(N/K))$ rows with high probability. We remark that both these results are optimal up to a $\log K$ factor~\cite{DR82}.
\begin{thm}[\cite{DR82,PR08}]\label{thm:explicitdisjunct}
    There exist explicit constructions of $K$-disjunct matrices with $T=O(K^2\log N)$ rows. Moreover, it is possible to sample a $K$-disjunct matrix with $T=O(K^2\log(N/K))$ rows with high probability. Both these results are optimal up to a $\log K$ factor, and the matrix columns are $O(K\log N)$-sparse in the two constructions.
\end{thm}

\section{Bitmasked Matrices and Exact Sparse Recovery}\label{sec:sparserec}

In this section, we describe the bitmasking technique, along with its application in sparse recovery presented by Berinde et al.~\cite{BGIKS08}. As already mentioned, later on Cheraghchi and Indyk~\cite{CI17} modified this algorithm to handle approximate sparse recovery with stronger approximation guarantees, and gave a randomized version of this algorithm that allows for faster recovery.

Fix a matrix $W$ with dimensions $M\times N$. Consider another matrix $B$ of dimensions $\log N\times N$ such that the $j$-th column of $B$ contains the binary expansion of $j$ with the least significant bits on top. We call $B$ a \emph{bit-test matrix}. Given $W$ and $B$, we define a new \emph{bitmasked} matrix $W\otimes B$ with dimensions $M\log N\times N$ by setting the $i$-th row of $W\otimes B$ for $i=q\log N+t$ as the coordinate-wise product of the rows $W_{q\cdot}$ and $B_{t\cdot}$ for $q\in[M]$ and $t\in [\log N]$. This means that we have
\begin{equation*}
    (W\otimes B)_{i,j}= W_{q,j}\cdot B_{t,j} = W_{q,j}\cdot \mathsf{bin}_t(j),
\end{equation*}
for $i=q\log N+t$ and $j\in[N]$, where $\mathsf{bin}_t(j)$ denotes the $t$-th least significant bit of $j$.

Let $W$ be the adjacency matrix of a $(D,K,\eps)$-bipartite expander graph $G=(\cL=[N],\cR=[M],\cE)$. We proceed to give a high level description of the sparse recovery algorithm from~\cite{BGIKS08}. Suppose we are given access to $(W\otimes B)x$ for some unknown $K$-sparse vector $x$ of length $N$. Recall that our goal is to recover $x$ from this product very efficiently. A fundamental property of the bitmasked matrix $W\otimes B$ is the following: Suppose that for some $q\in[M]$ we have
\begin{equation}\label{eq:condsupp}
    \supp(x)\cap \supp(W_{q\cdot})=\{u\}
\end{equation}
for some $u\in [N]$. We claim that we can recover the binary expansion of $u$ directly from the $\log N$ products
\begin{equation*}\label{eq:prods}
    (W\otimes B)_{q\log N}\cdot x,(W\otimes B)_{q\log N+1}\cdot x, \dots, (W\otimes B)_{q\log N+\log N-1}\cdot x.
\end{equation*}
In fact, if~\eqref{eq:condsupp} holds, then for $i=q\log N+t$ it is the case that
\begin{equation*}
    (W\otimes B)_i \cdot x = \sum_{j=1}^N W_{q,j}\cdot \bin_t(j) \cdot x_j = \bin_t(u),
\end{equation*}
since $W_{q,j}\cdot x_j\neq 0$ only if $j=u$. In words, if $u$ is the unique neighbor of $q$ in $\supp(x)$, then the entries of $(W\otimes B)x$ indexed by $q\log N,\dots,q\log N+\log N-1$ spell out the binary expansion of $u$.

By the discussion above, if the edges exiting $\supp(x)$ in $G$ all had different neighbors in the right vertex set, we would be able to recover $x$ by reading off the binary expansion of the elements of $\supp(x)$ from the entries of $(W\otimes B)x$. However, if there are edges $(u,v)$ and $(u',v)$ for $u,u'\in\supp(x)$ in $G$, it is not guaranteed that we will recover $u$ and $u'$ as elements of $\supp(x)$. While such collisions are unavoidable, and thus we cannot be certain we recover $x$ exactly, the expander properties of $G$ ensure that the number of collisions is always a small fraction of the total number of edges. This means we will make few mistakes when reconstructing $\supp(x)$. As a result, setting $\eps$ to be a small enough constant and using a simple voting procedure (which we refrain from discussing as it is not relevant to us), we can recover a sparse vector $y$ that approximates $x$ in the sense that
\begin{eqn}
    \|x-y\|_0\leq \frac{\|x\|_0}{2}.
\end{eqn}
We can then repeat the procedure on input $(W\otimes B)(x-y)$ to progressively refine our approximation of $x$. 
As a result, we recover a $K$-sparse vector $x$ exactly in at most $\log K$ iterations.

\section{Code Construction and Decoding}\label{sec:code}

In this section, we define our code that is able to tolerate a prescribed number of errors, and analyze efficient syndrome decoding and full decoding algorithms.

Let $N$ be the desired block-length of the code, $K$ an upper bound on the number of adversarial errors introduced, and $\epsilon\in(0,1)$ a constant to be determined later. We fix an adjacency matrix $W$ of a $(D,K,\eps)$-layered unbalanced bipartite expander graph $G=(\cL,\cR,\cE)$ with $\cL=[N]$ and $\cR=[D\cdot M]$, where
\begin{equation*}
    D=O\left(\frac{\log N}{\eps}\right)\quad \textrm{and}\quad M=O\left(\frac{K}{\eps}\right).
\end{equation*}
Such an expander $G$ can be obtained with high probability by sampling a random function $C\colon [N]\times [D]\to [M]$ as detailed in Section~\ref{sec:expanderscondensers}. 

We define our code $\cC\subseteq \{0,1\}^N$ by setting its parity-check matrix $H$ as
\begin{equation*}
    H=\begin{bmatrix}
    W\\
    W\otimes B
    \end{bmatrix}.
\end{equation*}
It follows that $H$ has dimensions $(D\cdot M(1+\log N))\times N$, and so $\cC$ has rate at least
\begin{equation*}
    1-\frac{D\cdot M(1+\log N)}{N} = 1-O\left(\frac{K\cdot \log^2 N}{\eps^2 N}\right).
\end{equation*}
In particular, if $K$ and $\epsilon$ are constants, then $\cC$ has rate at least $1-O\left(\frac{\log^2 N}{N}\right)$.

Note that instead of storing the whole parity-check matrix $H$ in memory, one can just store the function table of $C$, which requires space $ND\log M$, along with a binary lookup table of dimensions $\log N\times N$ containing the $\log N$-bit binary expansions of $0,\dots,N-1$.

\subsection{Syndrome Decoding}\label{sec:syndec}

In this section, we study algorithms for syndrome decoding of $\cC$. Fix some codeword $c\in\cC$, and suppose $c$ is corrupted by some pattern of at most $K$ (adversarially chosen) errors. Let $x$ denote the resulting corrupted codeword. We have $x=c+e$, for a $K$-sparse error vector $e$. The goal of syndrome decoding is to recover $e$ from the syndrome $Hx$ as efficiently as possible. Our decoder is inspired by the techniques from~\cite{BGIKS08} presented in Section~\ref{sec:sparserec}, and also the sparse recovery algorithms presented in~\cite{CI17}. 

For the sake of exposition, we consider only the case of decoding over GF(2) -- the adaptation to arbitrary fields is simple and is discussed at the end of Section~\ref{sec:detsynrec}.

\subsubsection{A Deterministic Algorithm}\label{sec:detsynrec}

In this section, we present and analyze our deterministic decoder which on input $Hx$ recovers the error vector $e$ in sublinear time. 
Before we proceed, we fix some notation: For $s\in[D]$, let $G_s$ denote the subgraph of $G$
induced by the function $C_s$ (recall Definition~\ref{def:expander}), 
and let $W_s$ be its adjacency matrix. Informally, our decoder receives
\begin{equation*}
    Hx=\begin{bmatrix}
    Wx\\
    (W\otimes B)x
    \end{bmatrix}
\end{equation*}
as input, and works as follows:
\begin{enumerate}
    \item Estimate the size of $\supp(e)$. This can be done by computing $\|W_s\cdot x\|_0$ for all seeds $s\in[D]$ and taking the maximum;

    \item Using information from $W_{s}\cdot x$ and $(W_{s}\otimes B)x$ for the good seed fixed in Step $2$, recover a $K$-sparse vector $y$ which approximates the error pattern $e$ following the ideas from Section~\ref{sec:sparserec};
    
    \item If needed, repeat these three steps with $x-y$ in place of $x$.
\end{enumerate}
A detailed description of the deterministic decoder can be found in Algorithm~\ref{alg:detdecoder}. We will proceed to show the procedure detailed in Algorithm~\ref{alg:detdecoder} returns the correct error pattern $e$, provided $\epsilon$ is a small enough constant.
\begin{algorithm}
\caption{Deterministic decoder}\label{alg:detdecoder}
\begin{algorithmic}[1]
\Procedure{Estimate}{$Wx$}\Comment{Estimates $|\supp(e)|$ and outputs best seed}
\For{$s\in[D]$}
    \State Compute $L_s=\|W_s\cdot x\|_0$
\EndFor
\State Output $(\max_s L_s,\arg\max_s L_s)$
\EndProcedure

\Procedure{Approximate}{$W_s\cdot x,(W_s\otimes B)x$}\Comment{Computes a good approximation of $\supp(e)$}
\State Set $y=\mathbf{0}$

\For{$q=0,1,\dots,M-1$}
    \If{$(W_{s}\cdot x)_q\neq 0$}
        \State Let $u\in[N]$ be the integer with binary expansion
        \begin{equation*}
            (W_s\otimes B)_{q\log N}\cdot x,(W_s\otimes B)_{q\log N+1}\cdot x, \dots, (W_s\otimes B)_{q\log N+\log N-1}\cdot x,
        \end{equation*}
        \State Set $y_u=1$.
    \EndIf
\EndFor
\State Output $y$
\EndProcedure

\Procedure{Decode}{$Hx$} \Comment{The main decoding procedure}

\State Set $(L,s)=\textsc{Estimate}(Wx)$

\If{$L=0$}\Comment{If no errors found, stop and return the zero vector}
    \State Output $y=\mathbf{0}$

\Else

    \State Set $y=\textsc{Approximate}(W_s\cdot x,(W_s\otimes B)x)$

    \State Set $z=\textsc{Decode}(H(x-y))$
    \State Output $y+z$

\EndIf

\EndProcedure

\end{algorithmic}
\end{algorithm}

We begin by showing that procedure \textsc{Estimate} in Algorithm~\ref{alg:detdecoder} returns a good approximation of the size of $\supp(e)$ along with a good seed.
\begin{lem}\label{lem:estimatesupport}
    Procedure $\textnormal{\textsc{Estimate}}(Wx)$ in Algorithm~\ref{alg:detdecoder} returns a seed $s\in[D]$ satisfying
    \begin{equation*}
        (1-2\epsilon)|\supp(e)|\leq |\Gamma_s(\supp(e))|\leq |\supp(e)|,
    \end{equation*}
    where the parameter $\eps$ comes from the underlying expander graph.
\end{lem}
\begin{proof}
    Recall that we defined $L_s=\|W_s\cdot x\|_0$ and $L=\max_{s\in[D]}\|W_{s}\cdot x\|_0$.
    First, observe that $W_s\cdot x=W_s\cdot e$ for all $s$. Then, the upper bound $L\leq |\supp(e)|$ follows from the fact that
    \begin{equation*}
        \|W_s\cdot e\|_0\leq |\supp(e)|
    \end{equation*}
    for all seeds $s$ and all vectors $e$, since each column of $W_s$ has Hamming weight $1$ by the fact that the underlying graph $G$ is layered. 
    
    It remains to lower bound $L$. Since $e$ is $K$-sparse, we know that
    \begin{equation*}
        |\Gamma(\supp(e))|\geq (1-\epsilon)D|\supp(e)|.
    \end{equation*}
    By an averaging argument, it follows there is at least one seed $s\in[D]$ such that
    \begin{equation*}
        |\Gamma_s(\supp(e))|\geq (1-\epsilon)|\supp(e)|.
    \end{equation*}
    As a result, by Lemma~\ref{lem:singleneighbor} the number of vertices in $\Gamma_s(\supp(e))$ adjacent to a single vertex in $\supp(e)$ is at least
    \begin{equation*}
        (1-2\epsilon)|\supp(e)|,
    \end{equation*}
    and thus $|\Gamma_s(\supp(e))|\geq \|W_s\cdot x\|_0\geq (1-2\epsilon)|\supp(e)|$.
\end{proof}

We now show that, provided a set $\cX\subseteq \cL=[N]$ has good expansion properties in $G_s$, then procedure \textsc{Approximate} in Algorithm~\ref{alg:detdecoder} returns a good approximation of $\cX$.
\begin{lem}\label{lem:recover}
    Fix a vector $x\in\{0,1\}^N$ and denote its support by $\cX$. 
    Suppose that
    \begin{equation}\label{eq:suppneighbor}
        |\Gamma_s(\cX)|\geq \left(1-2\eps\right)|\cX|
    \end{equation}
    for a given seed $s\in[D]$.
    Then, procedure $\textnormal{\textsc{Approximate}}(W_s\cdot x,(W_s\otimes B)x)$ returns an $|\cX|$-sparse vector $y\in\{0,1\}^N$ such that
    \begin{equation*}
        \|x-y\|_0\leq 5\eps \|x\|_0.
    \end{equation*}
\end{lem}
\begin{proof}
    First, it is immediate that the procedure returns an $|\cX|$-sparse vector $y$. This is because $|\Gamma_s(\cX)|\leq |\cX|$ as $G_s$ is $1$-left regular, and the procedure adds at most one position to $y$ per element of $\Gamma_s(\cX)$.
    
    In order to show the remainder of the lemma statement, observe that we can bound $\|x-y\|_0$ as
    \begin{equation*}
        \|x-y\|_0\leq A+B,
    \end{equation*}
    where $A$ is the number of elements of $\cX$ that the procedure does not add to $y$, and $B$ is the number of elements outside $\cX$ that the procedure adds to $y$.

    First, we bound $A$. Let $\Gamma'_s(\cX)$ denote the set of neighbors of $\cX$ in $G_s$ that are adjacent to a single element of $\cX$. Then, the lower bound in~\eqref{eq:suppneighbor} and Lemma~\ref{lem:singleneighbor} ensure that
    \begin{equation*}
        |\Gamma'_s(\cX)|\geq \left(1-4\eps\right)|\cX|.
    \end{equation*}
    Note that, for each right vertex $q\in \Gamma'_s(\cX)$, the bits
    \begin{equation*}
        (W\otimes B)_{q\log N}\cdot x,(W\otimes B)_{q\log N+1}\cdot x, \dots, (W\otimes B)_{q\log N+\log N-1}\cdot x
    \end{equation*}
    are the binary expansion of $u$ for a distinct $u\in \cX$, and $u$ is added to $\supp(y)$. As a result, we conclude that
    \begin{equation*}
        A\leq 4\eps|\cX|.
    \end{equation*}
    
    It remains to bound $B$.
    Note that the procedure may only potentially add some $u\not\in \cX$ if the corresponding right vertex $q$ is adjacent to at least three elements of $\cX$. 
    This is because right vertices $q$ that are adjacent to exactly two elements of $\cX$ satisfy $(W_s\cdot x)_q=0$, and so are easily identified by the procedure and skipped.
    Then, the lower bound in~\eqref{eq:suppneighbor} ensures that there are at most
    \begin{equation*}
        \frac{1}{2}\cdot 2\eps|\cX|=\eps|\cX|
    \end{equation*}
    such right vertices. As a result, we conclude that
    \begin{equation*}
        B\leq\eps|\cX|,
    \end{equation*}
    and hence
    \begin{equation*}
        \|x-y\|_0\leq 5\eps|\cX|.\qedhere
    \end{equation*}
\end{proof}

We combine the lemmas above to obtain the desired result.
\begin{coro}\label{coro:decdet}
   Suppose that $\epsilon<1/10$.
   Then, on input $Hx$ for $x=c+e$, $\textnormal{\textsc{Decode}}(Hx)$ in Algorithm~\ref{alg:detdecoder} recovers the error vector $e$ after at most $1+\frac{\log K}{\log\left(\frac{1}{5\eps}\right)}$ iterations.
\end{coro}
\begin{proof}
    Under the conditions in the corollary statement, combining Lemmas~\ref{lem:estimatesupport} and~\ref{lem:recover} guarantee that in the first iteration of \textsc{Decode} we obtain a $K$-sparse vector $y_1$ such that
    \begin{equation*}
        \|e-y_1\|_0\leq 5\eps \|e\|_0\leq 5\eps K.
    \end{equation*}
    Recursively applying this result shows that after $\ell$ iterations we have a vector $y=y_1+y_2+\cdots+y_\ell$ with sparsity at most $\sum_{i=1}^\ell (5\eps)^{i-1}K\leq 2K$ satisfying
    \begin{equation*}
        \|e-y\|_0\leq (5\eps)^\ell K.
    \end{equation*}
    Setting $\ell=1+\frac{\log K}{\log\left(\frac{1}{5\eps}\right)}$ ensures that $\|e-y\|_0<1$, and hence $e=y$.
\end{proof}

To conclude this section, we give a detailed analysis of the runtime of the deterministic syndrome decoder. We have the following result.
\begin{thm}\label{thm:runtime}
    On input $Hx$ for $x=c+e$ with $c\in\cC$ and $e$ a $K$-sparse error vector, procedure $\textnormal{\textsc{Decode}}(Hx)$ in Algorithm~\ref{alg:detdecoder} returns $e$ in time
    \begin{equation*}
       O\left(\frac{\log K}{\log\left(\frac{1}{5\eps}\right)}(K+M)(D+\log N)\right).
    \end{equation*}
    In particular, if $\epsilon$ is constant, $M=O(K/\eps)$, and $D=O(\log N/\epsilon)$, procedure \textnormal{\textsc{Decode}} takes time
    \begin{equation*}
        O(K\log K\cdot \log N).
    \end{equation*}
\end{thm}
\begin{proof}
    We begin by noting that, for a $K$-sparse vector $y$ and seed $s\in[D]$, we can compute $W_s\cdot y$ and $(W_s\otimes B)y$ in time $O(K)$ and $O(K\log N)$, respectively, with query access to the function table of $C$ and the lookup table of binary expansions. We now look at the costs incurred by the different procedures. We will consider an arbitrary iteration of the algorithm. In this case, the input vector is of the form $x+y$, where $x$ is the corrupted codeword and $y$ is a $2K$-sparse vector.
    \begin{itemize}
        \item The procedure \textsc{Estimate} in Algorithm~\ref{alg:detdecoder} requires computing $D$ products of the form $W_s(x+y)$, which in total take time $O(D(K+M))$, along with computing the $0$-norm of all resulting vectors. Since $W_s (x+y)$ has length $M$, doing this for all seeds takes time $O(DM)$. In total, the \textsc{Estimate} procedure takes time $O(D(K+M))$;

        \item The procedure \textsc{Approximate} in Algorithm~\ref{alg:detdecoder} requires the computation of $W_s(x+y)$ and $(W_s\otimes B)(x+y)$ for a fixed seed $s$, which take time $O(K+M)$ and $O((K+M)\log N)$, respectively. The remaining steps can be implemented in time $O(M\log N)$ for a total time of $O((K+M)\log N)$;
    \end{itemize}
    Note that the time required to compute the sum of sparse vectors in Line $19$ is absorbed into the $O((K+M)(D+\log N))$ complexity of previous procedures. The desired statements now follow by noting that there are at most $1+\frac{\log K}{\log\left(\frac{1}{5\eps}\right)}$ iterations.
\end{proof}

\begin{remark}
    In the proof of Theorem~\ref{thm:runtime}, we assume that reading an integer from memory (e.g., from the support of a sparse vector $y$ or the function table of $C$) takes time $O(1)$. If instead we assume that reading an $L$-bit integer from memory takes time $O(L)$, then we obtain runtime
    \begin{equation*}
       O\left(\frac{\log K}{\log\left(\frac{1}{5\eps}\right)}(K\log M+M)(D+\log N)\right).
    \end{equation*}
    instead.
\end{remark}

We conclude this section by briefly describing how to adapt Algorithm~\ref{alg:detdecoder} to perform sparse recovery over any field. There are several ways to do this. One possibility is to replace the check in Line $8$ of Algorithm~\ref{alg:detdecoder} by the following: $(W_s\cdot x)_q\neq 0$ \emph{and} all non-zero entries of
\begin{equation*}
    (W_s\otimes B)_{q\log N}\cdot x,\dots, (W_s\otimes B)_{q\log N+\log N-1}\cdot x
\end{equation*}
equal $(W_s\cdot x)_q$. As a result, right vertices in $\Gamma_s(\supp(e))$ with exactly two neighbors in $\supp(e)$ are skipped by the algorithm. Observe that this additional condition is trivially satisfied over $\GF(2)$ if $(W_s\cdot x)_q\neq 0$. Then, in Line $17$ one should set $y_u=(W_s\cdot x)_q$ instead. 

We remark that the condition in Line $8$ could be simplified in general to only checking whether $(W_s\cdot x)_q\neq 0$, at the expense of obtaining worse constants in the lemmas from this section. In particular, we would have to make $\epsilon$ smaller. 
Since we care about the practicality of our algorithms, we made an effort to have $\epsilon$ be as large as possible.

\subsubsection{A Randomized Algorithm}\label{sec:randsynrec}

In this section, we analyze a randomized version of Algorithm~\ref{alg:detdecoder} that is considerably faster. 
The main idea behind this version is that in procedure \textsc{Estimate} we can obtain a good estimate of $|\supp(e)|$ with high probability by relaxing $\epsilon$ slightly and sampling $\|W_s\cdot x\|_0$ for a few i.i.d.\ random seeds only.

In Algorithm~\ref{alg:randdecoder}, we present the randomized decoder, which uses an extra slackness parameter $\delta$. We will show that the procedure detailed in Algorithm~\ref{alg:randdecoder} returns the correct error pattern $e$ with probability at least $1-\eta$, provided $\epsilon$ is a small enough constant depending on $\delta$.
\begin{algorithm}
\caption{Randomized decoder}\label{alg:randdecoder}
\begin{algorithmic}[1]
\Procedure{Estimate}{$r,Wx$}\Comment{Estimates $|\supp(e)|$}
\State Sample $r$ i.i.d.\ random seeds $s_1,\dots,s_r$ from $[D]$
\State For each $s_i$, compute $L_i=\|W_{s_i}\cdot x\|_0$
\State Output $(\max_{i\in\{1,\dots,r\}} L_i,\arg\max_{i\in\{1,\dots,r\}} L_i)$
\EndProcedure

\Procedure{Approximate}{$W_s\cdot x,(W_s\otimes B)x$}\Comment{Computes a good approximation of $\supp(e)$}
\State Set $y=\mathbf{0}$

\For{$q=0,1,\dots,M-1$}
    \If{$(W_{s}\cdot x)_q\neq 0$}
        \State Let $u\in[N]$ be the integer with binary expansion
        \begin{equation*}
            (W_s\otimes B)_{q\log N}\cdot x,(W_s\otimes B)_{q\log N+1}\cdot x, \dots, (W_s\otimes B)_{q\log N+\log N-1}\cdot x,
        \end{equation*}
        \State Set $y_u=1$
    \EndIf
\EndFor
\State Output $y$
\EndProcedure

\Procedure{Decode}{$Hx,\eta,\delta,\epsilon$} \Comment{The main decoding procedure}

\State Set $r=1+\frac{\log(1/\eta)+\log\log K-\log\log\left(\frac{1}{5\eps(1+\delta)}\right)}{\log(1+\delta)}$

\State Set $(L,s)=\textsc{Estimate}(r,Wx)$

\If{$L=0$}
    \State Output $y=\mathbf{0}$

\Else

    \State Set $y=\textsc{Approximate}(W_s\cdot x,(W_s\otimes B)x)$

    \State Set $z=\textsc{Decode}(H(x-y),\eta,\delta,\epsilon)$
    \State Output $y+z$

\EndIf

\EndProcedure

\end{algorithmic}
\end{algorithm}

We begin by showing that procedure \textsc{Estimate} in Algorithm~\ref{alg:randdecoder} returns a good approximation of the size of $\supp(e)$ with high probability.
\begin{lem}\label{lem:estimatesupportrand}
    Procedure $\textnormal{\textsc{Estimate}}(r,Wx)$ in Algorithm~\ref{alg:randdecoder} returns a seed $s\in[D]$ satisfying
    \begin{equation}\label{eq:estimate}
        (1-2(1+\delta)\epsilon)|\supp(e)|\leq |\Gamma_s(\supp(e))|\leq |\supp(e)|
    \end{equation}
    with probability at least $1-(1+\delta)^{-r}$.
\end{lem}
\begin{proof}
    Similarly to the proof of Lemma~\ref{lem:estimatesupport}, the desired result will follow if we show that with probability at least $1-(1+\delta)^{-r}$ over the choice of the seeds $s_1,\dots,s_r$ it holds that
    \begin{equation}\label{eq:LsLB}
        L_s\geq (1-2(1+\delta)\epsilon)|\supp(e)|
    \end{equation}
    for some seed $s\in\{s_1,\dots,s_r\}$. 
    We note that~\eqref{eq:LsLB} holds if $|\Gamma_s(\supp(e))|\geq (1-(1+\delta)\epsilon)|\supp(e)|$. The probability that a random seed fails to satisfy this condition is at most $\frac{1}{1+\delta}$ by Lemma~\ref{lem:markov}. Therefore, the probability that none of the seeds $s_1,\dots,s_r$ satisfy the condition is at most $(1+\delta)^{-r}$, as desired.
\end{proof}

After we obtain a good seed via the \textsc{Estimate} procedure, invoking Lemma~\ref{lem:recover} with $\eps(1+\delta)$ in place of $\eps$ ensures that we can obtain a good sparse approximation of $\supp(e)$ with high probability.
We thus obtain the following corollary.
\begin{coro}
    Suppose that $\eps(1+\delta)<1/10$.
    Then, on input $Hx$ for $x=c+e$, procedure \textnormal{\textsc{Decode}} in Algorithm~\ref{alg:randdecoder} returns the error vector $e$ with probability at least $1-\eta$ in at most $1+\frac{\log K}{\log\left(\frac{1}{5\eps(1+\delta)}\right)}$ iterations.
\end{coro}
\begin{proof}
    The statement follows by repeating the proof of Corollary~\ref{coro:decdet} but replacing Lemma~\ref{lem:estimatesupport} with Lemma~\ref{lem:estimatesupportrand} and by invoking Lemma~\ref{lem:recover} with $\eps(1+\delta)$ in place of $\eps$.
    Since each iteration succeeds with probability at least $1-(1+\delta)^{-r}$ and there are at most $1+\frac{\log K}{\log\left(\frac{1}{5\eps(1+\delta)}\right)}$, a union bound guarantees that decoding fails with probability at most
    \begin{equation*}
        \left(1+\frac{\log K}{\log\left(\frac{1}{5\eps(1+\delta)}\right)}\right)\cdot (1+\delta)^{-r}\leq \eta
    \end{equation*}
    by the choice of $r$ in Algorithm~\ref{alg:randdecoder}.
\end{proof}

We conclude this section by analyzing the runtime of the randomized decoder.
\begin{thm}\label{thm:runtimerand}
    On input $Hx$ for $x=c+e$ with $c\in\cC$ and $e$ a $K$-sparse error vector, procedure \textnormal{\textsc{Decode}} Algorithm~\ref{alg:randdecoder} returns $e$ with probability at least $1-\eta$ in time
    \begin{equation*}
       O\left(\frac{\log K}{\log\left(\frac{1}{5\eps(1+\delta)}\right)}(K+M)\left(r+\log N\right)\right),
    \end{equation*}
    where
    \begin{equation*}
        r=1+\frac{\log(1/\eta)+\log\log K-\log\log\left(\frac{1}{5\eps(1+\delta)}\right)}{\log(1+\delta)}.
    \end{equation*}
\end{thm}
\begin{proof}
    The proof is analogous to that of Theorem~\ref{thm:runtime}, except that in the procedure \textsc{Estimate} in Algorithm~\ref{alg:randdecoder} we only test $r$ seeds.
        This means procedure \textsc{Estimate} now takes time $O\left(r(K+M)\right)$.
\end{proof}

We note that the runtime of the randomized decoder in Theorem~\ref{thm:runtimerand} is independent of the degree $D$ of the expander. This has two advantages: First, it means the hidden constants in the runtime are considerably smaller than in the deterministic case from Theorem~\ref{thm:runtime}, even assuming we use a near-optimal expander with degree $D=O(\frac{\log N}{\epsilon})$. Second, it means that replacing the near-optimal non-explicit expander graph by an explicit construction with sub-optimal parameters will affect the runtime of the randomized decoder only marginally. Furthermore, the failure probability $\eta$ only affects lower order terms of the runtime complexity. Therefore, we can (for example) set $\eta$ to be any arbitrarily small constant with only negligible effect in the runtime for large block lengths. Finally, we observe that computing the $0$-norm of vectors can be sped up with a randomized algorithm. One can simply sample several small subsets of positions and estimate the true $0$-norm with small error and high probability by averaging the $0$-norm over all subsets. As mentioned before, we will consider instantiations of our code with an explicit expander in Section~\ref{sec:instantiation}.

\begin{remark}
    As in the proof of Theorem~\ref{thm:runtime}, we assume in this section that reading an integer from memory takes time $O(1)$. If instead we assume that reading an $L$-bit integer from memory takes time $O(L)$, then we obtain runtime
    \begin{equation*}
       O\left(\frac{\log K}{\log\left(\frac{1}{5\eps(1+\delta)}\right)}\left(r\cdot (K(\log M+\log D)+M)+(K+M)\log N\right)\right)
    \end{equation*}
    instead. The preceding arguments still stand, as even for explicit expanders it holds that $\log M$ and $\log D$ are negligible compared to $\log N$.
\end{remark}

\subsection{Full Decoding}\label{sec:fulldecoding}

In this section, we study the decoding complexity of our code in the setting where we only have access to the corrupted codeword $x=c+e$. This mean that if we want to perform syndrome decoding, we must compute the parts of the syndrome that we want to use from $x$.

Recall that we have access to the function table of $C$, as well as a lookup table of $\log N$-bit binary expansions. As a result, we can compute products of the form $W_s\cdot x$ and $(W_s\otimes B)x$ in time $O(N)$ and $O(N\log N)$, respectively. This is because all columns of $W_s$ (resp.\ $W_s\otimes B$) have $1$ nonzero entry (resp.\ at most $\log N$ nonzero entries), and the nonzero entry (resp.\ entries) of the $j$-th column are completely determined by $C_s(j)=C(s,j)$ and the $j$-th column of the lookup table of binary expansions.

We now analyze the runtimes of both the deterministic and randomized decoders from Algorithms~\ref{alg:detdecoder} and~\ref{alg:randdecoder} in this alternative setting. We have the following results.
\begin{thm}\label{thm:runtimefull}
    On input $x=c+e$ with $c\in\cC$ and $e$ a $K$-sparse error vector, procedure \textnormal{\textsc{Decode}} from Algorithm~\ref{alg:detdecoder} returns $e$ in time
    \begin{equation*}
       O\left(\frac{\log K}{\log\left(\frac{1}{5\eps}\right)}(N+K+M)(D+\log N)\right).
    \end{equation*}
    In particular, if $\epsilon$ is constant, $M=O(K/\eps)$, and $D=O(\log N/\epsilon)$, procedure \textnormal{\textsc{Decode}} takes time
    \begin{equation*}
        O(K\log K\cdot N\log N).
    \end{equation*}
\end{thm}
\begin{proof}
    The proof is analogous to that of Theorem~\ref{thm:runtime}, except we now must take into account the time taken to compute products of the form $W_s\cdot x$ and $(W_s\otimes B)x$:
    \begin{itemize}
        \item The procedure \textsc{Estimate} in Algorithm~\ref{alg:detdecoder} requires computing $D$ products of the form $W_s(x+y)$, which in total take time $O(D(N+K+M))$, along with computing the $0$-norm of all resulting vectors. Since $W_s (x+y)$ has length $M$, doing this for all seeds takes time $O(DM)$. In total, the \textsc{Estimate} procedure takes time $O(D(N+K+M))$;

        \item The procedure \textsc{Approximate} in Algorithm~\ref{alg:detdecoder} requires the computation of $W_s(x+y)$ and $(W_s\otimes B)(x+y)$ for a fixed seed $s$, which take time $O(N+K+M)$ and $O((N+K+M)\log N)$, respectively. The remaining steps can be implemented in time $O(M\log N)$ for a total time of $O((N+K+M)\log N)$.
    \end{itemize}
    The desired statements now follow by noting that there are at most $1+\frac{\log K}{\log\left(\frac{1}{5\eps}\right)}$ iterations.
\end{proof}

\begin{thm}\label{thm:runtimerandfull}
    On input $x=c+e$ with $c\in\cC$ and $e$ a $K$-sparse error vector, procedure \textnormal{\textsc{Decode}} from Algorithm~\ref{alg:randdecoder} returns $e$ with probability at least $1-\eta$ in expected time
    \begin{equation*}
       O\left(\frac{\log K}{\log\left(\frac{1}{5\eps(1+\delta)}\right)}(N+K+M)\left(r+\log N\right)\right),
    \end{equation*}
    where
    \begin{equation*}
        r=1+\frac{\log(1/\eta)+\log\log K-\log\log\left(\frac{1}{5\eps(1+\delta)}\right)}{\log(1+\delta)}.
    \end{equation*}
\end{thm}
\begin{proof}
    The proof is analogous to that of Theorem~\ref{thm:runtimefull}, except that in the procedure \textsc{Estimate} in Algorithm~\ref{alg:randdecoder} we only need to test $r$ seeds.
        This means procedure \textsc{Estimate} now takes time $O\left(q(N+K+M)\right)$.
\end{proof}

There are important properties that are not explicit in the proof of Theorem~\ref{thm:runtimerandfull}. Observe that only one computation takes $O(N\log N)$ per iteration of the Algorithm~\ref{alg:randdecoder}; Namely, the computation of $(W_s\otimes B)x$ for a fixed seed $s$. Consequently, the hidden constant in the computation time is small. 
Moreover, as already discussed for the randomized syndrome decoder, the runtime is independent of the degree of the expander, and the effect of the failure probability on the runtime is negligible for large block lengths.
This means that the decoder described in Algorithm~\ref{alg:randdecoder} is also faster in the full decoding setting than the $O(ND)$ expander codes decoder adapted from~\cite{JXHC09}, whose running time depends on the degree of the expander graph, as long as the number of iterations is not too large. 
This is the case if the number of errors $K$ allowed is a small constant (e.g., $K\leq 5$) and we set $\eps$ to be not too large. Furthermore, observe that, unlike our decoder, the runtime of the expander codes decoder is affected by a sub-optimal choice of unbalanced bipartite expanders.
Finally, if we want a faster decoder for an arbitrary but fixed error threshold $K$, we can also set $\eps$ to be small enough so that the maximum number of iterations is sufficiently small for our needs. 
In this case, the rate of our code becomes smaller since we must make $\epsilon$ smaller.

\subsection{Instantiation with Explicit Expanders}\label{sec:instantiation}

In this section, we analyze how instantianting our construction with an explicit layered unbalanced expander with sub-optimal parameters affects the properties of our codes. 

More precisely, we consider instantiating our code with the GUV expander introduced by Guruswami, Umans, and Vadhan~\cite{GUV09}, and an explicit highly unbalanced expander constructed by Ta-Shma, Umans, and Zuckerman~\cite{TUZ07}. For simplicity, in this section we will assume that all parameters not depending on $N$ (such as $K$ and $\epsilon$) are constants. 

Fix constants $\alpha,\epsilon,K>0$. Then, the GUV graph is a $(D,K,\eps)$-layered bipartite expander with degree
\begin{equation*}
    D=O\left(\frac{\log N\cdot \log K}{\eps}\right)^{1+1/\alpha},
\end{equation*}
and, for each layer, a right vertex set of size
\begin{equation*}
    M=D^2\cdot K^{1+\alpha}.
\end{equation*}
Observe that, although the GUV expander is unbalanced, the size of its right vertex set grows with the degree. Ta-Shma, Umans, and Vadhan~\cite{TUZ07} provided explicit constructions of highly unbalanced layered expanders. In particular, they give a construction of a $(D,K,\eps)$-layered bipartite expander with degree
\begin{equation*}
    D=2^{O(\log\log N)^3},
\end{equation*}
and, for each layer, a right vertex set of size
\begin{equation*}
    M=K^{O(1/\eps)}.
\end{equation*}

Plugging the parameters of both graphs presented in this section into the runtimes in Theorems~\ref{thm:runtime},~\ref{thm:runtimerand},~\ref{thm:runtimefull}, and~\ref{thm:runtimerandfull} and treating $K$, $\eps$, and $\alpha$ as constants immediately yields explicit high-rate codes with syndrome and full decoding complexity displayed in Table~\ref{table:syndec} (for syndrome decoding), and in Table~\ref{table:fulldec} (for full decoding). 

\begin{table}[h]
\centering
\begin{tabular}{|c|c|c|}
\hline
Syndrome decoding             & Deterministic               & Randomized               \\ \hline
Graphs from \cite{GUV09} & $O(\log N)^{3+3/\alpha}$ & $O(\log N)^{3+2/\alpha}$ \\ \hline
Graphs from \cite{TUZ07} & $2^{O(\log\log N)^3}$       & $O(\log N)$              \\ \hline
\end{tabular}
\caption{Complexity of deterministic and randomized syndrome decoding for different explicit graphs when the number of errors $K$ and expander error $\eps$ are constants.}\label{table:syndec}
\end{table}

\begin{table}[h]
\centering
\begin{tabular}{|c|c|c|}
\hline
Full decoding             & Deterministic               & Randomized               \\ \hline
Graphs from \cite{GUV09} & $O(N\log^{1+1/\alpha} N)$ & $O(N\log N)$ \\ \hline
Graphs from \cite{TUZ07} & $O(N 2^{O(\log\log N)^3})$       & $O(N\log N)$             \\ \hline
\end{tabular}
\caption{Complexity of deterministic and randomized full decoding for different explicit graphs when the number of errors $K$ and expander error $\eps$ are constants.}\label{table:fulldec}
\end{table}

Observe that for both graphs there is a substantial decrease in complexity for randomized decoding versus deterministic decoding for the same setting. 
This is due to the fact that the decoding complexity of our randomized decoding algorithms is independent of the degree of the underlying expander, which we have already discussed before, and that the degree of explicit constructions is sub-optimal. 
Using the highly unbalanced explicit graphs from~\cite{TUZ07}, the decoding complexity of our randomized algorithms essentially matches that of the case where we use a random expander with near-optimal parameters (we are ignoring the contribution of $K$, which we assume to be small).

We conclude by noting that the full decoding complexity of expander codes under the explicit graphs from this section matches the second column of Table~\ref{table:fulldec}. In comparison, our randomized algorithm performs better under both graphs.

\section{Group Testing}\label{sec:grouptesting}

In this section, we show how we can easily obtain a scheme for non-adaptive group testing with few tests and sublinear time recovery. More precisely, we will prove the following:
\begin{thm}
    Given $N$ and $K$, there is an explicit test matrix $W$ of dimensions $T\times N$, where $T=O(K^2\log^2 N)$, such that it is possible to recover a $K$-sparse vector $x$ from $W\odot x$ in time $O(K^3\log^2 N)$.
\end{thm}

We begin by describing the test matrix $W$. Let $W'$ be an explicit $K$-disjunct matrix of dimensions $M\times N$ with $M=O(K^2\log N)$. Such explicit constructions exist as per Theorem~\ref{thm:explicitdisjunct}. Then, our test matrix $W$ is defined as
\begin{equation*}
    W=\begin{bmatrix}
    W'\\
    W'\otimes B
    \end{bmatrix},
\end{equation*}
where $B$ is the $\log N\times N$ bit-test matrix from Section~\ref{sec:sparserec}. It follows immediately that $W$ has dimensions $T\times N$ with $T=M\log N=O(K^2\log^2 N)$.

It remains to describe and analyze the recovery algorithm that determines $x$ from
\begin{equation*}
    W\odot x=\begin{bmatrix}
    W'\odot x\\
    (W'\otimes B)\odot x
    \end{bmatrix}
    =
    \begin{bmatrix}
    y^{(1)}\\
    y^{(2)}
    \end{bmatrix},
\end{equation*}
whenever $x$ is $K$-sparse. At a high-level, the algorithm works as follows:
\begin{enumerate}
    \item For $q\in[M]$, let $s_q$ be the integer in $[N]$ with binary expansion
    \begin{equation*}
        y^{(2)}_{q\log N},y^{(2)}_{q\log N+1},\dots,y^{(2)}_{q\log N+\log N-1}.
    \end{equation*}
    Recover the (multi) set $\cS=\{s_0,s_1,\dots,s_{M-1}\}$. The disjunctness property of the underlying matrix $W'$ ensures that $\supp(x)\subseteq \cS$;
    
    \item Similarly to the original recovery algorithm for disjunct matrices, run through all $s\in\cS$ and check whether $W'_s$ is contained $y^{(1)}$. Again, the fact that $W'$ is disjunct ensures that this holds if and only if $s\in\supp(x)$.
\end{enumerate}
A rigorous description of this recovery procedure can be found in Algorithm~\ref{alg:nagt}. We now show that procedure $\textnormal{\textsc{Recover}}(W\odot x)$ indeed outputs $\supp(x)$, provided that $x$ is $K$-sparse. First, we prove that $\supp(x)\subseteq \cS$.
\begin{lem}\label{lem:superset}
    Suppose that $x$ is $K$-sparse. Then, if $\cS=\textnormal{\textsc{SuperSet}}(W\odot x)$, we have $\supp(x)\subseteq \cS$.
\end{lem}
\begin{proof}
    It suffices to show that if $\supp(x)=\{s_0,\dots,s_{t-1}\}$ for $t\leq K$, then there is $q\in[M]$ such that $W'_{q,s_0}=1$ and $W'_{q,s_j}=0$ for all $1\leq j\leq t-1$. If this is true, then
    \begin{equation*}
        y^{(2)}_{q\log N},y^{(2)}_{q\log N+1},\dots,y^{(2)}_{q\log N+\log N-1}
    \end{equation*}
    would be the binary expansion of $s_0$. The desired property follows since $W'$ is $K$-disjunct. In fact, if this was not the case, then $W'_{\cdot s_0}$ would be contained in $\bigvee_{i=1}^{t-1} W'_{\cdot s_i}$, and hence $W'$ would not be $K$-disjunct.
\end{proof}
The next lemma follows immediately from the fact that $W'$ is $K$-disjunct.
\begin{lem}\label{lem:rfp}
    If $x$ is $K$-sparse and $\supp(x)\subseteq \cS$, then $\textnormal{\textsc{Remove}}(W\odot x,\cS)$ returns $\supp(x)$.
\end{lem}
Combining Lemmas~\ref{lem:superset} and~\ref{lem:rfp} with Algorithm~\ref{alg:nagt} leads to the following result.
\begin{coro}
   If $x$ is $K$-sparse, then $\textnormal{\textsc{Recover}}(W\odot x)$ returns $\supp(x)$.
\end{coro}

We conclude this section by analyzing the runtime of procedure $\textnormal{\textsc{Recover}}(W\odot x)$. We have the following result.
\begin{thm}
    On input a $K$-sparse vector $x$, procedure $\textnormal{\textsc{Recover}}(W\odot x)$ returns $\supp(x)$ in time $O(K^3\log^2 N)$.
\end{thm}
\begin{proof}
    We analyze the runtime of procedures SuperSet and Remove separately:
    \begin{itemize}
        \item In procedure $\textnormal{\textsc{SuperSet}}(W\odot x)$, for each $q\in[M]$ we need $O(\log N)$ time to compute $s_q$ and add it to $\cS$. It follows that $\textnormal{\textsc{SuperSet}}(W\odot x)$ takes time $O(M\log N)=O(K^2\log^2 N)$;
        
        \item In procedure $\textnormal{\textsc{Remove}}(W\odot x,\cS)$, it takes time $O(|\supp(W'_{\cdot j})|)$ to decide whether $W'_{\cdot j}$ is contained in $W'\odot x=y^{(1)}$. Therefore, in total the procedure takes time $O(|\cS|\cdot \max_{j\in\cS}|\supp(W'_{\cdot j})|)$. Noting that $|\cS|\leq M$ and that $|\supp(W'_{\cdot j})|=O(K\log N)$ for all $j$ implies that the procedure takes time
        \begin{eqn}
            O(M\cdot K\log N)=O(K^3\log^2 N).\qedhere
        \end{eqn}
    \end{itemize}
\end{proof}

\begin{algorithm}
\caption{Recovery algorithm for non-adaptive group testing scheme}\label{alg:nagt}
\begin{algorithmic}[1]

\Procedure{SuperSet}{$W\odot x$}\Comment{Recover superset of $\supp(x)$}

\State Set $\cS=\{\}$

\For{$q=0,1,\dots,M-1$}
    \State Let $s_q$ be the integer with binary expansion
    \begin{equation*}
        y^{(2)}_{q\log N},y^{(2)}_{q\log N+1},\dots,y^{(2)}_{q\log N+\log N-1}
    \end{equation*}
    \State Add $s_q$ to $\cS$
\EndFor

\State Output $\cS$

\EndProcedure

\Procedure{Remove}{$W\odot x,\cS$} \Comment{Finds all false positives in the superset and removes them}

\For{$s\in\cS$}
    \If{$W'_{\cdot s}$ is not contained in $y^{(1)}$}
        \State Remove $s$ from $\cS$
    \EndIf
\EndFor

\State Output $\cS$

\EndProcedure

\Procedure{Recover}{$W\odot x$}\Comment{Main recovery procedure}

\State Set $\cS=\textnormal{SuperSet}(W\odot x)$

\State Output $\textnormal{Remove}(W\odot x,\cS)$

\EndProcedure

\end{algorithmic}
\end{algorithm}

\section*{Acknowledgments}
The authors thank Shashanka Ubaru and Thach V. Bui for discussions on the role of the bitmasking technique in sparse recovery.
M.\ Cheraghchi's research was partially supported by the National Science Foundation under Grants No.\ CCF-2006455 and CCF-2107345.
J.\ Ribeiro's research was partially supported by the NSF grants CCF-1814603 and CCF-2107347, the NSF award 1916939, DARPA SIEVE program, a gift from Ripple, a DoE NETL award, a JP Morgan
Faculty Fellowship, a PNC center for financial services innovation award, and a Cylab seed funding award.

\bibliographystyle{IEEEtran}
\bibliography{refs}

\end{document}